\documentclass[runningheads]{llncs}

\usepackage[linesnumbered,ruled,vlined]{algorithm2e}
\usepackage{tikz}
\usetikzlibrary{matrix,positioning,arrows.meta,arrows,patterns,hobby}
\usepackage{graphicx}
\usepackage{pgfplots}
\begin{document}
\title{Secured Distributed Algorithms Without Hardness Assumptions}
%\titlerunning{Secured Distributed Algorithms without Hardness Assumptions}

%\author{Anonymous}{-}{}{}{}

\author{Leonid Barenboim \inst{1}\and Harel Levin \inst{1,2}}

\institute{
Department of Mathematics and Computer Science, The Open University of Israel, P.O.B. 808, Ra'anana, Israel \\
\and Department of Physics, Nuclear Research Center-Negev, P.O.B. 9001, Be'er-Sheva, Israel}

\maketitle      

\begin{abstract}
We study algorithms in the distributed message-passing model that produce secured output, for an input graph $G$. Specifically, each vertex computes its part in the output, the entire output is correct, but each vertex cannot discover the output of other vertices, with a certain probability. This is motivated by high-performance processors that are embedded nowadays in a large variety of devices. Furthermore, sensor networks were established to monitor physical areas for scientific research, smart-cities control, and other purposes. In such situations, it no longer makes sense, and in many cases it is not feasible, to leave the whole processing task to a single computer or even a group of central computers. As the extensive research in the distributed algorithms field yielded efficient decentralized algorithms for many classic problems, the discussion about the security of distributed algorithms was somewhat neglected. Nevertheless, many protocols and algorithms were devised in the research area of secure multi-party computation problem (MPC or SMC). However, the notions and terminology of these protocols are quite different than in classic distributed algorithms. As a consequence, the focus in those protocols was to work for every function $f$ at the expense of increasing the round complexity, or the necessity of several computational assumptions. In this work, we present a novel approach, which rather than turning existing algorithms into secure ones, identifies and develops those algorithms that are inherently secure (which means they do not require any further constructions). This approach yields efficient secure algorithms for various locality problems, such as coloring, network decomposition, forest decomposition, and a variety of additional labeling problems. Remarkably, our approach does not require any hardness assumption, but only a private randomness generator in each vertex. This is in contrast to previously known techniques in this setting that are based on a pre-construction phase, shared randomness or public-key encryption schemes. 
\end{abstract}


\pagebreak
\section{Introduction}
Over the last few decades, computational devices get smaller and are embedded in a wide variety of products. High-performance processors are embedded in smart phones, wearable devices and smart home devices. Furthermore, sensor networks were established to monitor physical areas for scientific research, smart-cities control and  other purposes. In such situations, it no longer makes sense, and in many cases it is not feasible, to leave the whole processing task to a single computer or even a group of central computers. In the distributed algorithms research field, all the processors are employed to solve a problem together. The basic assumption is that all the processors run the same program simultaneously. The network topology is represented by a graph $G=(V,E)$ where each processor (also referred as \textit{node}) is represented by a vertex, $v \in V$. Each communication line between a pair of processors $v,u\in V$ in the network is represented by an edge $(v,u) \in E$. \par
The time complexity of a distributed algorithm is measured by rounds. Each round consists of three steps: (1) Each processor receives the messages that were sent by its neighbors on the previous round. (2) Each processor performs a local computation. (3) Each processor may send messages to its neighbors. The time complexity of distributed algorithms is measured by the number of rounds necessary to complete an algorithm. Local computations (that is, computations performed inside the nodes) are not taken into account in the running time analysis in this model. \par
Despite the extensive research in the distributed algorithms field in the last decades, the discussion about the security of distributed algorithms was somewhat neglected. Nevertheless, many protocols and algorithms were devised in the research area of cryptography and network security. The secure multi-party computation problem (MPC or SMC) is one of the main problems in the cryptography research. However, the notions and terminology of these protocols is quite different than in classic distributed algorithms. Moreover, most of these protocols assume the network forms a complete graph. Additionally, the protocols have no restriction on the amount of communication between the nodes. \par
In this work we devise secure distributed algorithms, in the sense that the output of each processor is not revealed to others, even though the overall solution expressed by all outputs is correct. Our notion of security is the following. Consider a problem where the goal is assigning a label to each vertex or edge of the graph $G = (V,E)$, out of a range $[t]$, for some positive $t$. A secure algorithm is required to compute a proper labeling, such that for any vertex $v \in V$, (respectively edge $e \in E$) the other vertices in $V$ (resp. edges in $E$) are not aware of the label of $v$ (resp. $e$). Moreover, other vertices or edges can guess the label with probability at most $1/\lambda$, for an appropriate parameter $\lambda \leq t$. Note that this requirement can be achieved if each participant $v$ (resp. $e$) in the network computes a set of labels $\{l^1,l^2,...,l^{\lambda}\}$ ($l^i \in [t]$), such that any selection from its set forms a proper solution, no matter which selections are made in the sets of other participants. For example, in a proper coloring problem, if each vertex computes a set of colors (rather than just one color), and the set is disjoint from the sets of all its neighbors, the goal is achieved. In this case each
participant draws a solution from its set of labels uniformly at random. The result is kept secret by the participant, and thus others can guess it with probability at most $1/\lambda$. Thus, if the number of labels is small, the possibility of guessing a result of a vertex becomes quite large, inevitably. As we will demonstrate later, one can artificially increase the amount of labels to achieve smaller probabilities. However, when it is impossible to use a large number of labels, other techniques can be taken into account (such as Parter and Yogev's compiler \cite{parter2017distributed}). Nevertheless, our method is applicable to various distributed problems. Moreover, the overhead caused by the privacy preserving is negligible as the round complexity of our algorithms is similar to the best known (non privacy preserving) algorithms. A summary is found in Table \ref{table:results}. The parameter $\lambda$ is referred to as the \textit{solution domain} in Table \ref{table:results}. The ratio between $t$ and $\lambda$ is referred to as the \textit{contingency factor}. These terms will be discussed later in Chapter \ref{chap:inheren}. \par

\begin{table}
    \caption{List of inherently secure algorithms and their privacy attributes}
    \label{table:results}
    \resizebox{1\textwidth}{!}{
    \begin{tabular}{lcccc}
        %\toprule
        
         \textbf{Problem} & \textbf{Type of Graph} & \textbf{Rounds Complexity} & \textbf{Solution Domain Size} & \textbf{Contingency Factor} \\
         %\midrule
         
         $3\Delta$-Coloring & Oriented trees & $O(\log^*{n})$ & $\Delta$ & 3 \\
         $2c\cdot \Delta\log{n}$-Coloring & General & $O(1)$ & $c\cdot\log{n}/2$ & $O(\Delta)$ \\
         $O(\Delta^2)$-Coloring & General & $\log^*{n}+O(1)$ & $\Delta$ & $O(\Delta)$ \\
         $p$-Defective $O\left(\left(\frac{\Delta}{p}\right)^2\right)$-Coloring & General & $O(\log^*{n})$ & $O\left(\frac{\Delta}{p}\right)$ & $O\left(\frac{\Delta}{p}\right)$ \\
         $2a\cdot c\cdot \log{n}$-Coloring & Bounded Arboricity $a$ & $O(\log{n})$ & $O(\log{n})/2$ & $O(a)$ \\
         
         %\midrule
         $(O(\log{n}), O(c\cdot\log{n}))$-Network Decomposition & General & $O(\log^2{n})$ & $c>1$ & $O(\log{n})$ \\
         
         %\midrule
         
         $\Delta$-Forest Decomposition & General & $O(1)$ & $\Delta\choose {|E(v)|}$ & 1 \\
         $(2+\epsilon)\cdot a$-Forest Decomposition & Bounded Arboricity $a$ & $O(\log{n})$ & ${((2+\epsilon)\cdot a)}\choose {|E(v)|}$ & 1 \\
         
         %\midrule
         
         $O(\Delta\log{n})$-Edge Coloring & General & $O(1)$ & $c\cdot\log{n}$ & $O(\Delta)$ \\
         $O(\Delta^2)$-Edge Coloring & General & $\log^*{(n)}+O(1)$ & $(2\Delta-1)$ & $O(\Delta)$ \\
         $p$-Defective $O\left(\left(\frac{\Delta}{p}\right)^2\right)$-Edge Coloring & General & $O(1)$ & $O\left(\left(\frac{\Delta}{p}\right)^2\right)$ & 1 \\
         
         %\midrule
         
         $(t\cdot \sqrt{\Delta})$-Edge Coloring of a Dominating Set & General & $\tilde{O}(\log{\Delta}+\log^3{\log{n}})$ & $t$ & $\sqrt{\Delta}$ \\
         
         %$\alpha$-Matching & General & $O(1)$ & $\Delta \choose \alpha$ & 1 \\
         
         %\bottomrule
    \end{tabular}
    }
    \vspace{-10pt}
\end{table}

\section{Background}
\subsection{Distributed Algorithms}
Given a network of $n$ processors (or \textit{nodes}), consider a graph $G=(V,E)$ such that $V={v_1,v_2,...,v_n}$ is a set of vertices, each represents a processor. For each two vertices $u,v \in V$, there is an edge $(u,v) \in E$ if and only if the two processors corresponding to the vertices $u,v$ have a communication link between them. A communication link may be unidirectional or bidirectional, resulting in an undirected or a directed graph (respectively). Unless stated otherwise, the graphs in this work are simple, undirected and unweighted. \par
Two vertices $u, v \in V$ are \textit{independent} if and only if $(u,v) \notin E$. The \textit{neighbors} set of a vertex $v \in V$, $\Gamma(v)$ consists of all the vertices in $V$ that share a mutual edge with $v$ in $E$. Formally, $\Gamma(v) = \{u\in V | (u,v) \in E\}$. The \textit{degree} of a vertex $v \in V$, $deg(v) = |\Gamma(v)|$. Note that $0 \leq deg(v) \leq n-1$. The \textit{maximum degree} of graph $G$, $\Delta(G)$, is the degree of the vertex $v \in V$ which has the maximum number of neighbors. If the graph $G$ is directed, the out (respectively, in) degree of vertex $v \in V$ ($deg_{out}(v)$ and resp. $deg_{in}(v)$) is the number of edges $(u,v) \in E$ ($u \in V$) with orientation that goes out from (respectively, in to) vertex $v$. \par
Throughout this paper, $\mathcal{LOCAL}$ model will be used as the message-passing model. In this model, each communication line can send at each round an unrestricted amount of bits. It means that the primary measure is the number of rounds each node needs to "consult" its neighborhood by sending messages. This is in contrast to $\mathcal{CONGEST}$ model, where the bandwidth on each communication line on each cycle is bounded by $O(\log{n})$. \par
A single bit can pass from one endpoint of the graph to the other endpoint in $D(G)$ rounds (where $D(G)$ is the diameter of graph $G$). Thus, in the $\mathcal{LOCAL}$ model we usually look for time complexity lower than $O(D(G))$ and even sub-logarithmic (in terms of $|V|$), since all the nodes can learn the entire topology of the graph in $O(D(G))$ rounds and then perform any computation on the entire graph. Consequently, the research in local distributed algorithms is focused on solving those graph theory problems which have solutions that depend on the local neighborhood of each vertex rather than the entire graph topology. \par
Most of the problems in graph theory may be classified into two types. The first type is a bipartition of the graph (whether the vertices, the edges, or both) into two sets. For some problems both of these sets are of interest, and for other problems only one of the sets, while the other sets may be categorized as "all the rest". Examples of such problems include Maximal Independent Set and Maximal Matching. The other type of problems partitions the graph into several sets. This type of problems may be referred as "labeling" problems where there is a set of valid labels and every part of the graph is labeled by a unique label. Examples of such problems include Coloring and Network Decomposition.
\par In the following sections we will present some of the problems in the field of graph theory which exploit the potential of distributed algorithms. While we discuss some of the main bipartition problems, in our model for privacy preserving the labeling problems are more relevant.

\subsection{Graph Theory Problems}
The definition of the problems is given here briefly, a detailed definition can be found in Appendix \ref{app:problem-def}. \par
A function $\varphi :V \rightarrow [\alpha]$ is a legal \textit{$\alpha$-Coloring} of graph $G=(V,E)$ if and only if, for each $\{v,u\}\in E \rightarrow \varphi(v) \neq \varphi(u)$. Similarly, a function $\varphi: E \rightarrow [\alpha]$ is a valid \textit{$\alpha$-edge coloring} of graph $G=(V,E)$ i.f.f. for any vertex $v \in V$ there are no two distinct vertices $u,w \in \Gamma(v)$ such that $\varphi((v,u)) = \varphi((v,w))$. While a deterministic construction of such $(\Delta+1)$-graph coloring requires at least $O(\log^*{n})$ rounds \cite{linial1987distributive}. A randomized $(\Delta+1)$-graph coloring can be done in $\mathrm{poly}(\log\log{n})$ rounds \cite{rozhovn2020polylogarithmic}. Graph coloring is of special interest due to its applications in many resource management algorithms. In particular, certain resource allocation tasks require a proper coloring (possibly of a power graph) and that each vertex knows its own color, but not the colors of its neighbors. For example, this is the case in certain variants of Time Division Multiple Access schemes. \par
A \textit{forest} is a graph which contains no cycles. A \textit{forest decomposition} of graph $G=(V,E)$ is an edge-disjoint partition of $G$, to $\alpha$ sub-graphs $\mathcal{F}_1,\mathcal{F}_2,...,\mathcal{F}_\alpha$ such that $\mathcal{F}_i$ is a forest 
for every $1 \leq i \leq \alpha$. One way to define the \textit{arboricity} of a graph $G$ is as the minimal number of forests which are enough to fully cover $G$. \par
Given a graph $G=(V,E)$ and a vertex-disjoint partition of graph $G=(V,E)$ to $\alpha$ clusters $\mathcal{C}_1,\mathcal{C}_2,...,\mathcal{C}_\alpha$, we define an auxiliary graph $\mathcal{G}=(\mathcal{V},\mathcal{E})$ such that $\mathcal{V} = \{\mathcal{C}_1,\mathcal{C}_2,...,\mathcal{C}_\alpha\}$ and $(\mathcal{C}_u,\mathcal{C}_v)\in \mathcal{E}$ ($\mathcal{C}_u,\mathcal{C}_v\in \mathcal{V}$) iff $\exists (u,v) \in E$ such that $u\in\mathcal{C}_u$ and $v\in\mathcal{C}_v$.
The partition $\mathcal{C}_1,\mathcal{C}_2,...,\mathcal{C}_\alpha$ is a valid \textit{$(\mathrm{d},\mathrm{c})$-network decomposition} \cite{awerbuch1989network} if (1) the chromatic number of $\mathcal{G}$ is at most $\mathrm{c}$ and (2) the distance between each pair of vertices contained in the same cluster $v,u\in \mathcal{C}_i$ is at most $\mathrm{d}$. In strong network decomposition, the distance is measured with respect to the cluster $\mathcal{C}_i$ (in other words, $dist_{\mathcal{C}_i}(v,u) \leq \mathrm{d}$). In weak network decomposition, the distance is measured with respect to the original graph $G$ (in other words, $dist_G(v,u) \leq \mathrm{d}$).
Different algorithms yield different kind of network decompositions which satisfy different values of $\mathrm{d}$ and $\mathrm{c}$. One of the most valuable decompositions, which presents good trade-off between the radius of each cluster and the chromatic number of the auxiliary graph is an $(O(\log{n}),O(\log{n})$-network decomposition. \par
A set $I \subseteq V$ of vertices is called an \textit{Independent Set} (\textit{IS}) if and only if for each pair of vertices $v,u \in I$ there is no edge $(v,u) \in E$. An independent set $I$ is \textit{Maximal Independent Set} (\textit{MIS}) i.f.f. there is no vertex $v \in V \setminus I$ such that $I \cup \{v\}$ is a valid independent set. Similarly, a set of edges $M \in E$ is called a \textit{Matching} i.f.f. there is no pair of vertices $u_1,u_2 \in V$ ($u_1 \neq u_2$) such that $\exists{v\in V}$ where $\{(u_1,v),(u_2,v)\} \subseteq M$. A matching $M$ is \textit{Maximal Matching} (\textit{MM}) i.f.f. there is no edge $e \in E \setminus M$ such that $M \cup \{e\}$ is a valid matching. \par

\subsection{Secure Multi-Party Computation}
\label{chap:bgmultiparty}
\textit{Multi-Party Computation} (\textit{MPC}) is the ability of a party consisting of $n$ participants to compute a certain function $f(x_1,x_2,...,x_n)$ where each participant $i (1 \leq i \leq n)$ holds only its own input $x_i$. At the research fields of cryptography and networks security, \textit{Secure MPC} \cite{yao1982protocols} protocols enables parties to compute a certain function $f$ without revealing their own input ($x_i$). Our security model is information-theoretic secure, which means it is not based on any computational assumptions. Furthermore, we base our security notion on the \textit{semi-honest model} (as was devised by \cite{goldreich1987play}) which means that there are possibly curious participants but no \textit{malicious adversary}. In other words, adversary participant can not deviate from the prescripted protocol. However, it may be curious, meaning it may run an additional computation in order to find out private data of another participants. Permitting the existence of malicious adversaries which may collude with $t$ nodes will necessitate the graph to be $(2t+1)$-connected for security to hold (as shown by \cite{parter2017distributed}), which may be not a feasible constraint. \par
Previous works on secure-MPC (\cite{yao1982protocols}, \cite{goldreich1987play}) do not state any assumptions on the nature of neither the function $f$ nor the interactions between the participants. As a consequence, the privacy preserving protocols devised during the past decades are generalized for any kind of mathematical function and not necessarily computation of graph features. Furthermore, each of the participants is assumed to be an equal part of the computation. As such, any pair of participants is assumed to have a private communication line of its own. Translating those protocols to distributed algorithms for graph theory problems, will require a complete graph representing the communication which may be different than the input graph of the problem. While this approach is applicable in many realistic networks and problems, general networks with non-uniform communication topology may benefit from efficient distributed algorithms for computations where the desired function $f$ is local. Other works (such as \cite{halevi2016secure} and \cite{halevi2017non}) are dedicated to general graphs. However, their goal was not to optimize the rounds complexity as the protocols created by their algorithms will require at least $O(n^2)$ rounds even for a relatively simple function $f$. Furthermore, their techniques require a heavy setup phase, and based on some computational assumptions. Several other works provide secure protocols for general or sparse graphs (\cite{boyle2013communication} \cite{garay2008almost} \cite{chandran2010improved}). However, the focus in those protocols was to work for every function $f$, at the expense of increasing the round complexity, or the necessity of several computational assumptions. \par
Recently, Parter and Yogev \cite{parter2017distributed} \cite{parter2019secure} suggested a new kind of privacy notion which they tailored to the $\mathcal{CONGEST}$ distributed model. This notion allows a perfect security, in the information-theoretic sense, against a semi-honest adversary. In their notion, the neighbors of each node $v$ construct a \textit{private neighborhood tree} throughout which they broadcast a shared randomness. This randomness is used in order to encrypt the private variable of each neighbor. The node receives these encrypted private variables $x_1,x_2,...,x_t$ $(t=|\Gamma(v)|)$ and performs its local computation $f(x_1,x_2,...,x_t)$. Let $\mathrm{OPT}(G)$ be the best depth possible
for private neighborhood trees. Parter and Yogev devised an algorithm which constructs such trees in $O(n+\Delta\cdot\mathrm{OPT}(G))$ rounds, where each tree has depth of $O(\mathrm{OPT}(G)\cdot\mathrm{polylog}(n))$ and each edge $e\in E$ is part of at most $O(\mathrm{OPT}(G)\cdot\mathrm{polylog}(n))$ trees. Using their notion one can turn any $r$-rounds algorithm into a secure algorithm with an overhead of $\mathrm{poly}(\Delta,\log{n})\cdot\mathrm{OPT}(G)$ rounds for each round. Furthermore, they showed that for a specific family of distributed algorithms (to which they referred as "simple"), the round overhead can be reduced to $\mathrm{OPT}(G)\cdot\mathrm{polylog}(n)$. Using their method they have devised a variety of both global and symmetry-breaking local algorithms. However, their notion requires a pre-construction phase. Parter and Yogev obtain this phase within nearly-optimal running time. Still, in the security notion of perfect security against a semi-honest adversary, a pre-construction phase is unavoidable. Another unavoidable requirement is that the graph is bi-connected. On the other hand, our security notion is different, and deals with output security. Consequently, there is no need for a pre-construction phase, nor for a bi-connected graph. \par

\section{Inherently Secure Distributed Algorithms}
\label{chap:inheren}
Most of the classic distributed algorithms models assume that each vertex is aware of its neighbors. This is the case in our paper as well. Usually each node does not have any additional input except for its own ID. The output of the algorithm is a set of labels where each label corresponds to each vertex. Throughout the current work, vertices IDs will not be considered as a private input. Formally, from the perspective of node $v \in V$, a classic distributed algorithm calculates a function $f_v(D_{\Gamma(v)}) = l_v$, where $l_v\in [l]$ (for some constant $l$) is vertex $v$'s label which was calculated based on the input messages ($D_{\Gamma(v)}$) came from vertex $v$'s neighbors ($\Gamma(v)$). From a global perspective, the algorithm computes: $f:G(V,E) \rightarrow [l]^n$. This work considers the following security notion: each vertex $v$ cannot infer the value of $l_u$ (such that $u \in V, u \neq v$) with a certain probability. Our model assumes that each node $v \in V$ holds a private randomness generator $r_v$. \par
As an example, consider an algorithm for $\Delta^2$ graph multicoloring of graph $G=(V,E)$ with maximum degree $\Delta=\Delta(G)$ which provides any vertex $v$ with a set of $\Delta$ \textit{valid} colors $\varphi(v)=\{x_1,x_2,...,x_\Delta\}$. By "valid" we mean that any of the colors in $\varphi(v)$ is not contained in any of $v$'s neighbors' sets, i.e. $x_i\notin \bigcup\limits_{u\in \Gamma(v)}{\varphi(u)}$ (for any $1\leq i \leq \Delta$). Using this kind of coloring, $v$ can \textit{privately} select a random color out of the $\Delta$ valid colors in $\varphi(v)$. Hence, the identity of the exact color of $v$ can be securely hidden from any of the other vertices in $G$. \par
We generalize the above idea as follows. Consider the following family of algorithms. Each algorithm $\Pi$ in the $\textit{Inherently-Secure}$ algorithms family $\mathcal{IS}$ consists of two stages: (1) Calculating a generic set of $k$ possible valid labels. (2) Randomly and privately (using the private randomness generator $r_v$), each node selects its final label. That is, the first stage of algorithm $\Pi$ (denoted by $\Pi_{generic}$) calculates the function: $f_1(G(V,E)) = \{\ell_{u_1},...,\ell_{u_n}\}$, where $\ell_{u_i} = \{l_i^1,l_i^2,...,l_i^k\}$ for any $u_i \in V$. Henceforth, $\Pi_{generic}$ will be referred as \textit{generic-algorithm}. The first stage can run without any additional security considerations, meaning any node may know the $\ell_{u_i}$ of other nodes. Later, we will show algorithms which satisfies even stronger security notion where the identity of $\ell_{u_i}$ is also kept secret. The second stage ($\Pi_{select}$) securely calculates the function $f_2:{[\ell]^k}^n \rightarrow [l]^n$. Overall, algorithm $\Pi$ indeed calculates $f=f_1 \circ f_2 : G(V,E) \rightarrow \{l_1,...,l_n\}$. Let $L$ be the ground set of valid labels from which the possible labels are being picked, i.e. for any $1 \leq i \leq n$ and $1 \leq j \leq k$, $l_i\in L$ and $l_i^j\in L$ .\par
By increasing the amount of possible values ($k$) we make the actual labels $\{l_1,l_2,...,l_n\}$ less predictable. However, in order to do so we may need to increase the ground set of the available labels. For instance, in graph coloring we may want to be able to produce $\Delta$ valid possible colors for each vertex. However, an increase of the amount of colors (to $\Delta^2$) may be necessary. On the other hand, one may want to minimize the size of the ground set since large ground sets may lead to trivial algorithms on one hand, and to a higher memory complexity on the other hand. \par
In order to analyze this kind of algorithms we define several parameters.
\begin{definition}
The size of the problem domain of a problem $\mathcal{P}$ solved by algorithm $\Pi$ which calculates the function $f : G(V,E) \rightarrow \{l_1,l_2,..,l_n\}$, where $l_i \in L$, is the number of valid labels for any $l_i$, i.e. $|L|$.
\label{def:problem-domain}
\end{definition}
\begin{definition}
The size of the solution domain of a generic algorithm $\Pi_{generic}$ which calculates the function $f_1:G(V,E) \rightarrow \{\{l_1^1,l_1^2,...,l_1^k\},...\{l_n^1,l_n^2,...,l_n^k\}\}$ is the minimal number of valid possible labels for any vertex, i.e. $k$.
\label{def:solution-domain}
\end{definition}
\begin{definition}
The contingency factor of generic algorithm $\Pi_{generic}$ used to solve problem $\mathcal{P}$ (as they defined on definitions \ref{def:problem-domain} and \ref{def:solution-domain}) is the ratio between the size of the problem domain $|L|$ (Def. \ref{def:problem-domain}) and the size of the solution domain $k$ (Def. \ref{def:solution-domain}) , i.e. $|L|/k$.
\label{def:contingency-factor}
\end{definition}
In order to clarify these definitions, consider the problem of $\Delta^2$-graph coloring. The size of the problem domain is $\Delta^2$. A generic algorithm that calculates $\Delta$ possible valid colors for each vertex will provide a solution domain of size $\Delta$. The contingency factor of this algorithm will be $\frac{\Delta^2}{\Delta}=\Delta$. \par
In many cases, the number of labels (i.e. the size of the problem domain) can be increased artificially by a factor $c>1$. This artificial increase will lead to an expansion of the problem domain by the same factor $c$. As a result, the contingency factor will remain the same. That is, the contingency factor is a property of the algorithm itself and not influenced by artificial increases. Small contingency factor indicates that most of the members of the problem domain are valid options on the solution domain, while the generic algorithm did not exclude those members from being considered as valid possible solutions. As such, small contingency factor indicates that the algorithm preserves better security by excluding only a small portion of possible solutions. Problems with small problem-domain will have even smaller solution domain which will lead to a contingency factor that is close to the original size of the problem domain. Therefore, finding a generic algorithm with good contingency factor for these problems is a complicated task. As a consequence, we will focus on finding generic algorithms for problems with relatively large problem-domain, i.e. labeling problems. For problems with small problem domain, other techniques (such as Parter and Yogev's compiler \cite{parter2017distributed}) should be considered.  \par
Note that even though a malicious node may interrupt the validity of the algorithm by picking a solution which is not part of its solution domain, this kind of intrusion will not affect the privacy of the algorithm. However, in our model the nodes are not malicious.

\subsection{Generic Algorithms for Graph-Coloring}
\label{section:generic_graph_coloring}
Considering the problem of graph coloring, we will focus on finding generic algorithms that will provide contingency factor of $\Delta$. This factor is optimal for general graphs, as we prove in Theorem \ref{theorem:optimal_facotr_coloring}.

\begin{theorem}
For any $\alpha$-coloring problem ($\alpha > \Delta$), and any generic algorithm $\Pi$, there is an infinite family of graphs such that their solution domain must be of size $O(\alpha/\Delta)$ at most. Hence, the contingency factor would be at least $\Omega(\Delta)$.
\label{theorem:optimal_facotr_coloring}
\end{theorem}

\begin{proof} Suppose for contradiction that there is a valid solution domain such that every vertex has more than $\alpha/\Delta$ valid options. Consider a graph $G=(V,E)$ with clique $C \subseteq V$ of size $|C| = \Delta+1$. Each vertex $v \in C$ has $\Delta$ neighbors, each of them has $\alpha/\Delta$ valid colors. But since $v$ and all its neighbors are part of the clique, each of them has a unique set of colors. It means that there are at least $(\Delta + 1)\cdot(\alpha/\Delta) > \alpha$ colors in the $\alpha$-coloring, which is a contradiction. 
\end{proof}

\subsubsection{Generic Algorithm for $3\Delta$-Coloring of Oriented Trees}
\label{section:secure_cole_vishkin}

A well known algorithm for 3-Coloring of oriented trees was devised by Cole and Vishkin \cite{cole1986deterministic}. The deterministic algorithm exploits the asymmetric relationship between a vertex $v$ and its parent (in the tree) $\pi(v)$ in order to get a valid coloring in $O(\log^*{n})$ rounds. \par

For any two integers $a$ and $b$, let $<a,b>$ be a tuple which can be represented in binary as the concatenation of the binary representations of $a$ and $b$. We can define a generic algorithm that runs Cole Vishkin's algorithm to get a valid coloring $\varphi : V \rightarrow [3]$. Later, each vertex will set its solution domain $\hat{\varphi}(v)$ as follows: $\hat{\varphi}(v) = \bigcup\limits_{0\leq i < \Delta}{< i, \varphi(v)>}$. As a result, we get a generic algorithm where each vertex has $\Delta$ valid colors. The validity of this algorithm is provided by the following Lemma:\par

\begin{lemma}
\label{lemma:valid_coloring_cv_expansion}
For any vertex $v\in V$, for every value $x \in \hat{\varphi}(v)$, $x$ is not a possible color for any other vertex $u \in \Gamma(v)$. 
\end{lemma}
\begin{proof} Suppose for contradiction that there exists a vertex $u \in \Gamma(v)$ such that $x \in \hat{\varphi}(u)$. Since, $\hat{\varphi}(v) = \bigcup\limits_{0\leq i < \Delta}{< i, \varphi(v)>}$, there exists a value $1 \leq i \leq \Delta$ such that $x = <i, \varphi(v)> = <i, \varphi(u)>$. Hence, $\varphi(v) = \varphi(u)$, which is a contradiction since $\varphi$ is a valid coloring as was proved by \cite{cole1986deterministic}. \end{proof}

Lemma \ref{lemma:valid_coloring_cv_expansion} leads to the following corollary: 
\begin{corollary}
Given a tree $T=(V,E,v)$, there exists a generic algorithm which provides any vertex $v\in V$ with a set of $\Delta$ possible valid colors
\end{corollary}

The generic algorithm described above uses a simple approach which achieves privacy by artificially increasing the size of both the problem domain and the solution domain accordingly. 
Another technique is to run an algorithm $d$ times in parallel. For coloring problems a good $d$ will probably be $\Delta$ (as was shown in Theorem \ref{theorem:optimal_facotr_coloring}). While this approach will lead in deterministic algorithms to the same results as the previous technique, applying this technique with random algorithms will lead to solution domain which is somewhat less predictable than the domain we will receive by artificially increasing the size of the solution domain. \par

While these approaches (artificially increasing the size of the solution domain and run the algorithm multiple times) are useful for problems with very efficient base algorithms (i.e. Cole Vishkin 3-coloring), for many problems such an efficient algorithm is not yet known. However, one still may devise efficient generic-algorithms for some of these problems, as demonstrated in the following sections.

\subsubsection{Generic Algorithm for $2\Delta c \cdot \log{n}$-Coloring of General Graphs}
\label{section:secure_coloring}
While the best known algorithms for $(\Delta+1)$-Coloring of generic graphs uses logarithmic number of rounds \cite{johansson1999simple} \cite{luby1993removing}, a reasonable size of contingency factor may yield more efficient algorithms with sub-logarithmic and even constant number of rounds. As an example, consider Algorithm \ref{algo:generic_coloring} which uses $O(1)$ rounds to achieve a secure coloring of general graphs with contingency factor of $O(\Delta)$. \par

\begin{algorithm}[H]\label{algo:generic_coloring}
\caption{GENERIC-RANDOM-COLORING}
\SetAlgoLined
\KwResult{A set of $O(\log{n})$ colors for each vertex $v \in V$ }
 Every vertex selects independently at random $k=c \cdot \log{n}$ different numbers ($c$ is a constant, $c>1$) $I=\{<1,x_1>,<2,x_2>,...,<k,x_k>\}$ where $ x_i \in [2 \Delta]$ is a number selected uniformally at random (for each $1 \leq i \leq k$).\\
 Send $I$ to each neighbor.\\
 For each message $\hat{I}=\{<1,\hat{x_1}>,<2,\hat{x_2}>,...,<k,\hat{x_k}>\}$ received, \textbf{do} $I \leftarrow I \setminus \hat{I}$. \\
 %Select privately, independently in random $x \in I$ and set it as the color.
\end{algorithm}

\vspace{5pt}
\begin{lemma} 
\label{lemma:generic_coloring_colors_left}  
For any set $I$ produced by algorithm \ref{algo:generic_coloring},  $|I| \geq k/2$ with high probability. 
\end{lemma}
\begin{proof}
For any color $x=<i,x_i>$ selected by $v \in V$, the probability that any of the neighbors picked any color which was selected by $v$ is (by union bound) $Pr\left[\exists x \in \bigcup \limits_{u\in \Gamma (v)}{\hat{I_u}} \right] \leq \Delta \cdot \frac{O(\log{n})}{2\Delta} \cdot \frac{1}{O(\log{n})}=\frac{1}{2}$. The probability that none of the neighbors picked any of $v$'s colors is therefore \\ 
$Pr \left[ \not \exists x
\in \bigcup \limits_{u\in \Gamma (v)}{\hat{I_u}} \right] = 1 -
Pr\left[x \in \bigcup \limits_{u\in \Gamma (v)}{\hat{I_u}} \right] \geq 0.5$. 
The expected size of $I$ is therefore $E\left[|I|\right] = k \cdot \left(1 - Pr\left[x \in \bigcup \limits_{u\in \Gamma (v)}{\hat{I_u}} \right] \right) \geq \frac{k}{2}$. 
The probability that the actual size of $I$ is lower than the expectation is given by Chernoff: $Pr[|I|<(1-\delta) \cdot k/2]=\left(\frac{e^{-\delta}}{(1-\delta)^{(1-\delta)}}\right)^{(k/2)}$. Since $k=O(\log{n})$, $Pr[|I|<(1-\delta) \cdot k/2] =\left(\frac{e^{-\delta}}{(1-\delta)^{(1-\delta)}}\right)^{(O(\log{n})/2)} \leq \frac{1}{n^c}$ (for a constant $c>1$). 
\end{proof}

\begin{lemma}
\label{lemma:generic_coloring_unique}
For any set $I_v$ produced by Algorithm \ref{algo:generic_coloring} running on vertex $v \in V$, and every color $i \in I_v$: $\forall{u \in \Gamma(v)},\hspace{3pt} i \notin I_u$.
\end{lemma}
\begin{proof}
Suppose that both $u$ and $v$ picked $i$ randomly at the first step of the algorithm. Hence, $i$ was in the set of colors $v$ sent to $u$ in the second step and vice-versa. Consequently, both $v$ and $u$ should remove $i$ from their possible colors list, i.e. $i \notin I_v$ and $i \notin I_u$. \end{proof}

The fact that Algorithm \ref{algo:generic_coloring} is privacy preserving is established by the Theorem \ref{theorem:generic_coloring_correctness} derived from Lemmas \ref{lemma:generic_coloring_colors_left} and \ref{lemma:generic_coloring_unique}. The contingency factor of the algorithm is $\frac{2\Delta c \cdot \log{n}}{c\cdot\log{n}/2}=O(\Delta)$.

\begin{theorem}
\label{theorem:generic_coloring_correctness} For any vertex $v \in V$, executing algorithm GENERIC-RANDOM-COLORING, the algorithm produces a set of at least $k/2$ valid colors in $O(1)$ rounds, with high probability. 
\end{theorem}
\subsubsection{Generic Algorithm for $O(\Delta^2)$-Coloring of General Graphs}
\label{section:poly-coloring}
Since currently known deterministic algorithms for $(\Delta + 1)$-coloring require at least $\sqrt{\log n}$ rounds, applying the simultaneous execution described above with such an algorithm will lead to relatively poor round complexity.
Instead, in this section we generalize a construction of \cite{barenboim2013distributed}\cite{linial1987distributive} which provides $O(\Delta^2)$-coloring in $\log^*{n} + O(1)$ rounds, in order to directly (i.e. without simultaneous executions) obtain secure algorithm for $O(\Delta^2)$-coloring. We employ a Lemma due to Erd{\"o}s et al. \cite{erdos1985families}. 

\begin{lemma}
\label{lemma:erdos}
For two integers $n$ and $\Delta$, $n > \Delta \geq 4$, there exists a family $\mathcal{J}$ of n subsets of the set $\{1,...,m\}$, $m=\lceil
\Delta^2 \cdot \ln{n} \rceil$, 
such that if $F_0,F_1,...,F_\Delta \in \mathcal{J}$ then $F_0 \not \subseteq \bigcup\limits_{i=1}^\Delta{F_i}$.
\end{lemma}
A set system $\mathcal{J}$ which satisfies the above is referred as \textit{$\Delta$-cover-free} set. \par
Erd{\"o}s et al. \cite{erdos1985families} also showed an algebraic construction which satisfies Lemma \ref{lemma:erdos}. For two integers $n$ and $\Delta$, using a ground set of size $m=O(\Delta^2\cdot \log^2{n})$, they construct a family $\mathcal{F}$ of $n$ subsets of the set $\{1,...,m\}$, such that $\mathcal{F}$ is a $\Delta$-cover-free. Linial \cite{linial1987distributive} showed that this construction can be utilized for distributed graph coloring. We construct a slightly different family which also provides multiple uncovered elements in each set:
\begin{theorem}
\label{theorem:algebraic_coloring_construction}
For two integers $n$ and $\Delta$, using a ground set of size $m=O(\Delta^2\cdot \log^2{n})$, there exists a family $\mathcal{F}$ of $n$ subsets of the set $\{1,...,m\}$, such that if $F_0,F_1,...,F_\Delta \in \mathcal{F}$ then $\left| F_0 \setminus \bigcup\limits_{i=1}^\Delta {F_i}\right| \geq \Delta$.
\end{theorem}
\begin{proof} Let $GF(q)$ be a field of characteristic $q$, where $q$ is a prime. Let $X=GF(q) \times GF(q)$, $|X|=m=q^2$, be a ground-set. For a positive parameter $d$, Poly$(d,q)$ is a set of all polynomials of degree $d$ over $GF(q)$, hence $\left| \mathrm{Poly}(d,q)\right|=q^{d+1}$. For each polynomial $g() \in $ Poly$(d,q)$, let $S_g=\{(a,g(a)) \mid a \in GF(q)\}$ be the set of all the points on the graph of $g()$ ($S_g \subseteq X$). Since $\forall{g()}$, $|S_g|=q$, and since for two distinct polynomials $g()$ and $g'()$, $S_g$ and $S_{g'}$ intersects in at most $d$ points, for a fixed set $S_g$ one needs at least $q/d$ other sets $S_h$ to cover $S_g$. \par
We construct a family $\mathcal{F}=\{S_g \mid g() \in \mathrm{Poly}(d,q)\}$, such that $\Delta = \lceil q/d \rceil -1$ . It follows that for each $S_{g_0},S_{g_1},...,S_{g_\Delta} \in \mathcal{F}$, $S_{g_0} \notin \bigcup \limits_{i=1}^\Delta {S_{g_i}}$. The cardinality of $\mathcal{F}$ is $|\mathcal{F}|=\left|\mathrm{Poly}(d,q)\right|=q^{d+1}=n$. It follows that $m\cdot \log^2{m} = 4q^2\log^2{q} \leq 4(\lceil q/d \rceil)^2(d+1)^2\log^2{q} = 4(\Delta+1)^2\log^2{n} \Rightarrow m \leq 4(\Delta+1)^2\log^2{n}$. When setting $d=2$, we get $n=q^3$, $m=q^2$. Setting $q \geq 3\Delta$, for each $\Delta+1$ different polynomials $F_0,...,F_\Delta$, we get $\left|F_0 \setminus \bigcup\limits_{i=1}^\Delta{F_i}\right| \geq \Delta$. 
\end{proof} 
These polynomials provides sets of labels such that if every vertex is assigned to a set, each set has at least $\Delta$ values which are not contained in any of its neighbors' sets. Figure \ref{fig:polynomials} demonstrates the logic behind the construction of $\Delta$-cover-free polynomials family $\mathcal{F}$. Each pair of second-degree polynomials intersect in at most two elements. Hence, for a ground set of 9 elements (as in figure \ref{fig:polynomials}), each polynomial and a set of 2 other polynomials leave at least $9-2\cdot 2 = 5$ unique elements which not intersect with any of the other polynomials. In figure \ref{fig:polynomials} even though the red polynomial intersects both the blue and the green polynomials in two different elements (each), we still have five different elements which were not intersected by any of the other polynomials. \par

\begin{figure}
    \centering

    \begin{tikzpicture}
  \begin{axis}[ 
    xlabel=$x$,
    xmin=-2,
    xmax=6,
    xtick={-2,-1,0,1,2,3,4,5,6}
  ] 
    \addplot [
    domain=-2:6, 
    samples=9, 
    color=red,
    mark=square
] {(x-2)*(x-3) -4}; 
    \addplot [
    domain=-2:6, 
    samples=9, 
    color=blue,
    mark=square
]{2*(x-2)*(x-4) -4}; 
    \addplot [
    domain=-2:6, 
    samples=9, 
    color=green,
    mark=square
]{-2*x^2+7*x+38}; 
  \end{axis}
\end{tikzpicture}
    \caption{Demonstration of construction of $2$-cover-free sets with at least $5$ uncovered elements.}
    \label{fig:polynomials}
\end{figure}
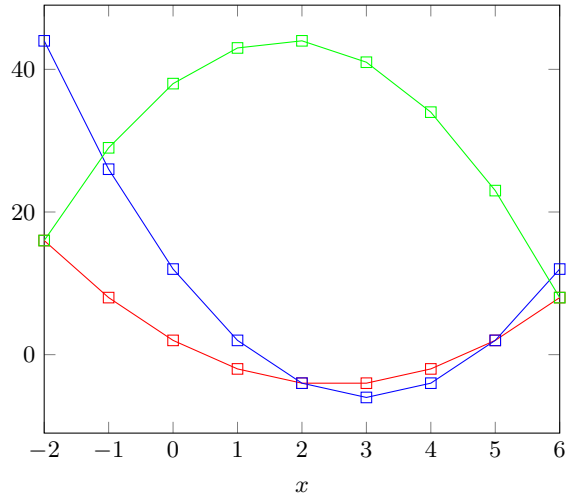 
\par

Next, we will use the constructions from \cite{erdos1985families} and Theorem \ref{theorem:algebraic_coloring_construction} to devise a generic-algorithm for $O(\Delta^2)$-coloring. Our algorithm is similar to Linial's iterative algorithm (\cite{linial1987distributive}), but instead of getting only one color on the last iteration, we get $\Delta$ different possible colors (for each vertex). \par
Starting with a valid $n$-coloring for some graph $G=(V,E)$ (the color of each vertex is its ID), we can apply the coloring algorithm from \cite{linial1987distributive} which will turn the $n$-coloring into an $O(\Delta^2\log^2{n})$-coloring in a single round. After $\log^*{n}+O(1)$ rounds we will get an $O(\Delta^2\log^2{\Delta})$-coloring.
%translate each color $c$ to a polynomial $g_c() \in \mathcal{F}$. Each vertex sends its polynomial to its neighbors, each vertex finds a point $(x,g_c(x))$ in its own polynomial which not intersects with any of its neighbors' polynomials. And set it as its new color. Hence, the number of colors was reduced from $n$ to $m=O(\Delta^2 \log^2{n})$. \par
%Applying the algorithm iteratively will achieve $O(\Delta^2 \log^2{\Delta})$-Coloring after $\log^*{n}$ rounds. \par
 For a sufficiently large $\Delta$, it holds that $O(\Delta^2\log^2{\Delta}) \leq (3\Delta)^3$. Hence, in order to further reduce the number of the colors to $O(\Delta^3)$ and get $\Delta$ valid optional colors we will use the set system from Theorem \ref{theorem:algebraic_coloring_construction} to reduce the $O(\Delta^2\log^2{\Delta})$-coloring to a $q^2 = (3\Delta)^2$-coloring of $G$ such that each vertex has at least $\Delta$ valid colors.  \par
To conclude, the problem domain is of size $9\Delta^2$. The solution domain contains of at least $\Delta$ valid colors. Consequently, the contingency factor is $O(\Delta)$, which is proved to be optimal (see Theorem \ref{theorem:optimal_facotr_coloring}).

\subsubsection{Generic Algorithm for $p$-Defective $O\left(\left(\frac{\Delta}{p}\right)^2\right)$-Coloring of General Graphs}
\label{section:poly-defective-coloring}
For a graph $G=(V,E)$, the function $\varphi: V \rightarrow [\alpha]$ is a valid $p$-defective $\alpha$-coloring iff for each vertex $v \in V$, the number of neighbors which have the same color as $v$ is at most $p$, i.e. $\left|\{u \in \Gamma(v) \mid \varphi(v)=\varphi(u)\} \right| \leq p$. \par
First, we shall find the lower bound of contingency factors for generic algorithms of defective coloring:
\begin{theorem}
\label{theorem:lower_bound_defective}
For any $p$-defective $\alpha$-coloring problem, there is an infinite family of graphs such that their solution domain must be of size $O\left(\alpha \cdot \frac{p}{\Delta}\right)$ at most and the contingency factor will be at least $\Omega\left(\frac{\Delta}{p}\right)$.
\end{theorem}
\begin{proof} Suppose for contradiction that there is a valid solution domain such that every vertex has more than $\alpha \cdot \frac{p}{\Delta}$ valid options. Consider a graph $G=(V,E)$ with a clique $C \subseteq V$ of size $|C| = \Delta + 1$. Each vertex $v \in C$ has $\Delta$ neighbors, each of them has $\alpha \cdot \frac{p}{\Delta}$. Since $v$ and all its neighbors are part of a clique, each color can be an optional color of at most $p$ different vertices. It means that there are at least $(\Delta + 1) \cdot (\alpha \cdot \frac{p}{\Delta})/p > \alpha$ colors in the $p$-defective $\alpha$-coloring, which is a contradiction. 
\end{proof}
The results from the previous section can be extended and combined with the results of \cite{barenboim2014distributed}, to achieve generic algorithm for $\rho$-defective $O\left(\left(\frac{\Delta}{\rho}\right)^2\right)$-coloring which provides a solution domain of size $O\left(\frac{\Delta}{\rho}\right)$. Such an algorithm provides an optimal contingency factor. \par
For a set $S_0$ and $\Delta$ other sets $S_1,...,S_\Delta$ and an integer $\rho > 0$, sets $S_1,...,S_\Delta$ $\rho$-cover $S_0$ if each element $x\in S_0$ is contained in at least $\rho$ of the sets $S_1,...,S_\Delta$. A family $\mathcal{F}$ is \textit{$\Delta$-union $(\rho +1)$-cover-free} iff for any $S_0,S_1,...,S_\Delta \in \mathcal{F}$, the sets $S_1,...,S_\Delta$ do not $(\rho+1)$-cover $S_0$. \par
Suppose we have a a $D$-defective $\alpha$-coloring $\varphi$ of graph $G$ where $\Delta(G)=\Delta$, and a $\Delta$-union $(\rho+1)$-cover-free family $\mathcal{F}$ with $\alpha$ sets over a ground set $[m]$, for some values $D,\alpha,\Delta,\rho$. Each vertex $v$ of $G$, in parallel, can find select an element $x\in S_{\varphi(v)}$ that belongs to at most $\rho$ sets in $S_{\varphi(u)}$ where $u\in \Gamma(v)$ such that $\varphi(u) \neq \varphi(v)$. The existence of such an element $x$ is guaranteed since $\mathcal{F}$ is a $\Delta$-union $(\rho+1)$-cover-free. Setting $v$'s new color $\varphi'(v) = x$, among the neighbors of $v$ there are at most $D$ neighbors $(u \in \Gamma(v))$ such that $\varphi(v) = \varphi(u)$, and another $\rho$ neighbors at most which had $\varphi(u) \neq \varphi(v)$ and selected the same element $x$. Hence, the new coloring, $\varphi'$ is a $(\rho + D)$-defective $m$-coloring. Consequently, Barenboim et al. state the following Lemma:
\begin{lemma} \textsc{\cite{barenboim2014distributed}}
For some values $D,\alpha,\Delta,\rho$ and $m$, and suppose we are given an $D$-defective $\alpha$-coloring of graph $G$ with $\Delta(G)=\Delta$ and a $\Delta$-union $(\rho+1)$-cover-free family $\mathcal{F}$ with $\alpha$ sets over a ground set $[m]$. Then, in one round one can compute a $(\rho + D)$-defective $m$-coloring of graph $G$.
\end{lemma}
Next we will show a construction which is based on the $\Delta$-union $(\rho+1)$-cover-free family devised by \cite{barenboim2014distributed}, but it slightly different since it fills one extra constraint that each set $S_0$ will contain at least $\frac{\Delta}{\rho+1}$ values which are not $(\rho+1)$-covered by any other $\Delta$ sets $(S_1,...,S_\Delta))$. \par
Consider the set of polynomials Poly($d,q$), and the family $\mathcal{F} = \{ S_g \mid g() \in \mathrm{Poly}(d,q) \}$ which were defined in section \ref{section:poly-coloring}. In order to hit each element of $X_0 \in \mathcal{F}$ for $\rho+1$ times, one needs at least $q\cdot (\rho+1)/d$ other polynomials. Let $\Delta = \lceil q \cdot (\rho+1)/d \rceil -1$, $m=q^2$ and $n=q^{d+1}$, where $\left|\mathcal{F}\right| = n$. It follows that $m\cdot\log^2{m}=4q^2\log^2{q} \leq 4\left(\frac{q}{d+1}\right)^2\cdot (d+1)^2\log^2{q} \leq 4\left(\frac{\Delta+1}{\rho+1}\right)^2\log^2{n} \Rightarrow m\leq 4\left(\frac{\Delta+1}{\rho+1}\right)^2\log^2{n}$. Consequently, one can use $\mathcal{F}$ in order to turn a valid $n$-coloring into $\rho$-defective $O\left(\frac{\Delta+1}{\rho+1}\right)^2\log^2{n}$-coloring.
Additionally, Setting $d=2$, we obtain $m=q^2, n=q^3$. Setting $\Delta \geq q\cdot(\rho+1)/2-1$, for $n \leq 8\cdot \left(\frac{\Delta+1}{\rho+1}\right)^3$ we get $m\leq 4\cdot \left(\frac{\Delta+1}{\rho+1}\right)^2)$. Furthermore, setting $q \geq 3\Delta/(\rho+1)$, there will be at least $\Delta/(\rho+1)$ elements in the set $S_0 \in \mathcal{F}$ that will be part of less than $\rho$ other sets for any other $S_1,...,S_\Delta \in \mathcal{F}$. By setting this system of polynomials, we can devise an algorithm that will prove Theorem \ref{theorem:poly_defective} \begin{theorem}
\label{theorem:poly_defective}
Given a graph $G=(V,E)$ $(|V|=n)$ with maximum degree $\Delta$, and a fixed parameter $1\leq p \leq \Delta$, there is a generic algorithm that calculates $p$-defective $O\left(\left(\frac{\Delta}{p}\right)^2\right)$-coloring with a solution domain of size $O(\Delta/p)$ and a contingency factor of at least $\Omega(\Delta/p)$, in $O(\log^*{n})$ rounds.
\end{theorem}
\begin{proof}
We start by applying the coloring algorithm from \cite{barenboim2014distributed} which will produce a $p$-defective $O\left(\left(\frac{\Delta+1}{p+1}\right)^2\right)$-coloring in $O(\log^*{n})$ rounds.

For a sufficiently large $\Delta$, it holds that $O(\Delta^2\log^2{\Delta}) \leq (3\Delta)^3$. Hence, if we use a field of characteristic $q \geq 3\Delta/(p+1)$ at the very last round, each vertex will have at least $O(\Delta/p)$ valid options for a $p$-defective $O\left(\frac{\Delta+1}{p+1}\right)^2$-coloring. Hence, we obtain a generic algorithm for a $p$-defective $O\left(\left(\frac{\Delta}{p}\right)^2\right)$-coloring with a solution domain of size $O(\Delta/p)$ and a contingency factor of at most $O(\Delta/p)$. 
\end{proof} 
The contingency factor is optimal by Theorem \ref{theorem:lower_bound_defective}.

\subsection{Generic Algorithm for Network Decomposition}
As was mentioned before, network decomposition may be referred as a labeling problem where the cluster IDs are the labels and the clusters assignment is the labeling function. Hence, network decomposition problems are good candidates for generic algorithms. However, considering the term of privacy in the network decomposition problem, different definitions may be suggested. One may suggest a permissive notion where all the members of the same cluster are allowed to share their private data with each other. This permissive notion makes sense since network decomposition is frequently used as a building block in other algorithms (such as coloring or finding MIS) where in the first stage each vertex discovers its cluster's topology, calculates private solution for the entire cluster, and then communicates with other clusters in order to generate an overall solution. However, even on a restrictive notion where each vertex may know only its own cluster assignment, some efficient algorithms may be suggested. Hence, during this work we will use the restrictive notion. \par
In section \ref{section:secure_cole_vishkin} we have shown how multiple simultaneous executions of the same random algorithm can expand the solutions domain and provide a generic algorithm for graph coloring problems. This approach can be adopted in order to expand the solution domain of network decomposition algorithms, However, since the network decomposition should satisfy certain constraints (namely, the depth of the clusters and the chromatic number of the auxiliary graph), this approach should be implemented carefully.
We will use the weak-diameter $(O(\log{n}),O(\log{n}))$-network decomposition random algorithm devised by Linial and Saks \cite{linial1993low}. This algorithm runs in $O(\log^2{n})$ rounds. Our generic algorithm proceeds as follows. Given a graph $G=(V,E)$, and a positive integer $c>1$, execute Linial and Saks's algorithm for $c$ times simultaneously, in parallel. Each of the execution will have its own serial number $i\in \{1,...,c\}$. Let $C_i:V \rightarrow \{1,...,O(\log{n})\}$ be a set of labeling functions ($1\leq i \leq c$), such that $C_i(v)=j$ iff vertex $v\in V$ was assigned by the $i$-th execution to cluster $\mathcal{C}_j$. For each vertex $v\in V$ we assign a set of $\log{n}$ different possible labels $C(v)=<1,C_1(v)>,...,<c,C_{c}(v)>$. Each of the labels will represent a distinct cluster ID. These labels are different since even if the independent executions produced the same cluster assignments, the first parameter on each tuple representing the label will be different since it represents the unique ID of each independent execution. \par
The clusters assignment described above is a privacy preserving $(O(\log{n}),O(c\cdot\log{n}))$-network decomposition. 
\begin{lemma}
For each of the possible labels, the weak diameter of each cluster is $O(\log{n})$ at most.
\end{lemma}
\begin{proof} Given a graph $G=(V,E)$ with two vertices $v,u$ assigned to the same cluster, i.e. there is a label $x$ such that $x\in C(u)$ and $x \in C(v)$. Hence, there was a specific execution ($i$) such that $x=<i,C_i(u)>=<i,C_i(v)>$, and therefore $C_i(u)=C_i(v)$. It means that on the $i$-th execution of Linial Saks's $(O(\log{n}), O(\log{n}))$-network decomposition, both $u$ and $v$ were part of cluster $C_i(u)=C_i(v)$. Consequently, by the proof of the correctness of Linial Saks's algorithm \cite{linial1993low}, the weak diameter of that cluster is at most $O(\log{n})$, i.e. there is a path between $u$ and $v$, $P_{u,v}=\{(u,w_1),(w_1,w_2),...,(w_k,v)\}\subseteq E$, of length at most $O(\log{n})$. 
\end{proof}
Each vertex $v\in V$ selects independently in random a cluster assignment $\mathcal{C}_k$ from $C(v)$. Then, 
we define an auxiliary graph $\mathcal{G}=(\mathcal{V},\mathcal{E})$ such that $\mathcal{V} = \{\mathcal{C}_1,...,\mathcal{C}_\alpha\}$ and $(\mathcal{C}_u,\mathcal{C}_v)\in \mathcal{E}$, ($\mathcal{C}_u,\mathcal{C}_v\in \mathcal{V}$) iff $\exists (u,v) \in E$ such that $u\in\mathcal{C}_u$ and $v\in\mathcal{C}_v$.
\begin{lemma}
For each of the possible labels, the auxiliary graph $\mathcal{G}$ has a chromatic number of $O(c\cdot\log{n})$ at most.
\end{lemma}
\begin{proof} Since there are, in total, $c \cdot O(\log{n}) = O(c\cdot\log{n})$ different cluster labels, any cluster can pick its label as its color. Since, the label of the clusters is unique, this is a legal vertex coloring in the auxiliary graph. Hence the chromatic number of $\mathcal{G}$ is at most $O(c\cdot\log{n})$. 
\end{proof} \par
The combination of the two Lemmas above, yields Theorem \ref{theorem:network-decomposition}.
\begin{theorem}
\label{theorem:network-decomposition}
Given a graph $G=(V,E)$, there is a generic algorithm which calculates weak-diameter $(O(\log{n}), O(c\cdot\log{n}))$-network decomposition in $O(\log^2{n})$ rounds. The algorithm produces $c$ valid possible cluster assignments for each vertex $v\in V$.
\end{theorem}

Since the size of the problem domain is $O(c\cdot\log{n})$ and the solution domain is of size $c$, the contingency factor is $O(\log{n})$. \par
Usually, it is useful to set the parameters of the network decomposition to be polylogarithmic in $n$. Hence, it may be useful to set $c=\log{n}$ and get an $(O(\log{n}), O(\log^2{n}))$-network decomposition. On the other hand, in order to preserve privacy, setting $c=\min(\Delta,\log{n})$ is sufficient as it allows each of the $\Delta$ neighbors of each vertex to have a different set of possible cluster assignments.

\subsection{Generic Algorithms for Forest Decomposition}
An \textit{oriented tree} is a directed tree $T=(V,E,r)$ where $r\in V$ is the root vertex, where every vertex $v \in V$ knows the identity of its parent $\pi(v)$ and has an oriented edge $(v,\pi(v))$. An \textit{oriented forest} is such a graph that any of its connected components are oriented trees. Any graph $G=(V,E)$ with maximum degree $\Delta$ can be decomposed into a set of $\Delta$ edge-disjoint forests $F_1,...,F_\Delta (F_i=(V_{F_i},E_{F_i})$ such that $E=\bigcup\limits_{1 \leq i \leq \Delta}{E_{F_i}}$. The problem of how to decompose a graph into forests can be viewed as a labeling problem where each edge should have a label $1 \leq i \leq \Delta$ that represents the forest $F_i$ which it belongs to. Since in every oriented forest, each vertex has at most 1 parent, each vertex will have at most $\Delta$ outgoing edges, each belongs to a different forest. Hence, for each vertex $v\in V$ there are $\Delta \choose \deg_{out}(v)$ different options to associate edges to different forests. From the edge's point of view, each of the $\Delta$ labels is a valid possible label. \par
Panconesi and Rizzi \cite{panconesi2001some} devised an algorithm for $\Delta$-forest decomposition of a general undirected graph in 2 rounds. Their algorithm can be viewed as two separate algorithms, each of a single round. The first algorithm is a simple yet powerful way to decompose a directed acyclic graph with maximum outgoing degree $d$ into $d$ oriented forests. The second is a a way to turn an undirected graph with maximum degree $\Delta$ into a directed acyclic graph with maximum outgoing degree $\Delta$. Each of these algorithms run in a single round. Combining these two algorithms produces a 2 round algorithm for $\Delta$-forest decomposition of any undirected graph with maximum degree $\Delta$. \par
Next we will describe Panconesi and Rizzi's algorithm for forest decomposition of oriented graphs. We will show how this algorithm can be modified in order to preserve privacy while maintaining a contingency factor of 1. Later, we will show two algorithms which produce a directed acyclic graphs. The first is Panconesi and Rizzi's algorithm for orienting any general undirected graph. The second algorithm (due to \cite{barenboim2010sublogarithmic}) performs an acyclic orientation for a graph with bounded arboricity $a$ such that the maximum outgoing degree is $\lfloor 2+\epsilon\rfloor\cdot a$. The combination of these algorithms yields a $\Delta$-forest decomposition for graphs with maximum degree $\Delta$ and $\lfloor 2 + \epsilon \rfloor \cdot a$-forest decomposition for graphs with bounded arboricity $a$. Both of them fit the constraint of preserving privacy.

\subsubsection{Forest Decomposition of Oriented Graphs}
Given a directed acyclic graph $G=(V,E)$, such that each vertex has a set of outgoing edges $E(v)=\{(v,u) \mid (v,u)\in E\}$ the single round algorithm of Panconesi and Rizzi \cite{panconesi2001some} goes as follows. Each vertex $v \in V$, in parallel, assigns a distinct number $1 \leq i \leq \left| E(v) \right|$ to each $e \in E(v)$. Let $\hat{E_i}$ be the set of all edges that were assigned with the number $i$. The forest decomposition is the set of forests $F_1,...,F_\Delta$ where $F_i = (V,\hat{E_i})$. The correctness of the algorithm was proved by \cite{panconesi2001some}. While in the original algorithm the nodes do not assign the labels randomly, in our algorithm a random assignment is required. Next, we analyze the privacy of the algorithm.  \par
\begin{theorem}
\label{theorem:forest-decomposition}
Panconesi and Rizzi's forest decomposition algorithm with random label assignment is privacy preserving and it has a contingency factor of 1. 
\end{theorem}
\begin{proof} 
Since the algorithm does not perform any kind of communication between the vertices, any kind of information cannot possibly leak from one vertex to another. As so, the forest decomposition algorithm has similar privacy characteristic as of choosing a color out of several valid colors. Since for every vertex $v \in V$ there are $\left| E(v) \right|$ outgoing edges, it has  $\left| E(v) \right|$ different parents, each on different forest. Therefore, any of the outgoing edges ($E(v)$) should be part of different forest on the forest decomposition. Hence, there are $\Delta \choose {\left|E(v)\right|}$ valid edge assignments, i.e. the size of the problem domain is $\Delta \choose {\left|E(v)\right|}$. However, any assignment which will assign any outgoing edge to a different forest will be valid, which mean the solution domain is of size $\Delta \choose {\left|E(v)\right|}$ too. Hence, the contingency factor is ${\Delta \choose {\left|E(v)\right|}}/{\Delta \choose {\left|E(v)\right|}} = 1$. \par
From the edge's point of view, each of the $\Delta$ labels is a valid possible label. Hence both the problem domain and the solution domain are of size $\Delta$. Henceforth, the contingency factor is 1.  
\end{proof}

\subsubsection{Acyclic Orientation of Graphs}
Panconesi and Rizzi \cite{panconesi2001some} showed that the simple orientation where each edge is oriented towards the vertex with the higher ID, is an acyclic orientation. Hence, any undirected graph can achieve an acyclic orientation in a single round, and can be decomposed privately into $\Delta$ forests in one additional round. This orientation provides each vertex with up to $\Delta!$ valid options for forest assignments. From the edge's point of view, each of the $\Delta$ labels is a valid possible label. Barenboim et al. \cite{barenboim2010sublogarithmic} devised an $O(\log{n})$-rounds algorithm that receives a graph with bounded arboricity $a$ and performs an acyclic orientation with maximum outgoing degree of $\lfloor 2 + \epsilon \rfloor \cdot a$. This orientation is achieved by partitioning the vertices of a graph $G$ into $l=\lfloor \frac{2}{\epsilon}\log{n} \rfloor$ sets $H_1,...,H_l$ such that each vertex $v\in H_i (i\in \{1,...,l\})$ has at most $(2+\epsilon)\cdot a$ neighbors in $\cup _{j=i}^l{H_j}$. Then, the orientation is done such that each edge $(u,v) \in E$ with endpoints $u\in H_i$ and $v\in H_j$, points towards the vertex that belongs to the higher ranked set (in case the two endpoints belong to the same set, the edge will point towards the vertex with the higher ID). This orientation provides each vertrex with up to ${(\lfloor 2 + \epsilon \rfloor \cdot a)}\choose {|E(v)|}$ valid options for forest assignments. From the edge's point of view, each of the $\lfloor 2+\epsilon\rfloor\cdot a$ labels is a valid possible label.
\subsection{Generic Algorithms for Graph Coloring of Graphs With Bounded Arboricity $a$}
The forest decomposition algorithms that was described above can be used as building blocks for other distributed algorithms for classic graph theory problems as graph coloring. In the following chapter we will use the $\lfloor 2 + \epsilon \rfloor \cdot a$-forest decomposition of \cite{barenboim2010sublogarithmic} that we showed in the previous chapter to achieve an $2a\cdot c \cdot \log{n}$-coloring for graphs with bounded arboricity $a$ (for any $c>1$). We will show that this coloring is private and has contingency factor of $O(a)$.\par

Combining the GENERIC-RANDOM-COLORING algorithm (algorithm \ref{algo:generic_coloring}) with the acyclic orientation algorithm devised by \cite{barenboim2010sublogarithmic} yields a secure algorithm for generic $2a\cdot c \cdot \log{n}$-Coloring for graphs with bounded arboricity $a$ such that from initial selection of $k=c \cdot \log{n}$ initial colors (for any $c>1$), each vertex has, at the end of the execution, at least $k/2$ valid optional colors. The algorithm basically performs the original random generic coloring, but it makes advantage of the acyclic orientation to break the symmetry between each pair of neighbors and make sure that only $O(a)$ neighbors constraint the valid residual colors of each vertex. \par
The algorithm consists of two steps. On the first step, the algorithm performs an acyclic orientation of graph $G$ such that the maximum outgoing degree is $\lfloor 2 + \epsilon \rfloor \cdot a$. The orientation is done by invoking the first two steps of procedure Forests-Decomposition (algorithm 2 in \cite{barenboim2010sublogarithmic}) with graph $G$ and parameter $0 < \epsilon \leq 2$. On the second step the generic coloring is done. Each vertex $v$ chooses independently at random $k=c \cdot \log{n}$ numbers from the range $[2\cdot A]$. These choices form a set of optional colors: $I_v=\{<1,x_1>,...,<k,x_k>\}$. Next, each vertex $v$ sends its set of colors $I_v$ to its children (in correspondence to the orientation). Each vertex $u$ which received a set $I_v$ from one of its parents performs $I_u \leftarrow I_u \setminus I_v$. \par
\begin{lemma}
The residual set of colors contains at least $k/2$ colors.
\end{lemma}
\begin{proof} The claim is derived directly from the proof of Lemma \ref{lemma:generic_coloring_colors_left}. 
\end{proof}
\begin{lemma}
For any vertex $v \in V$, there is no color $<i,x_i>$ in the residual available colors set $I_v$ such that $<i,x_i> \in \bigcup \limits_{u\in \Gamma(v)} {I_u}$.
\end{lemma} 
\begin{proof} Suppose for contradiction that $<i,x_i> \in I_u$ for some $u \in \Gamma(v)$. Let $F_j$ be the forest in $\mathcal{F}$ which includes the edge $(u,v)$. It means that either $u \in \pi_j(v)$ or $v \in \pi_j(u)$, which means that either $u$ or $v$ received it from its parent ($v$ or $u$, respectively) and should have removed it from its residual set, contradiction. 
\end{proof}
The time complexity of the algorithm follows from the time complexity of Procedure Forest-Decomposition($a,\epsilon$), which is $O(\log{n})$, plus $O(1)$ for coloring. The size of the problem domain is $2a\cdot c \cdot \log{n}$ and the size of the solution domain is $\frac{c\cdot\log{n}}{2}$. Hence, the contingency factor is $O(a)$.

\subsection{Generic Algorithms for Edge Coloring}
When considering the meaning of privacy in the context of edges, there is a slight difference between vertex coloring and edge coloring. Since the algorithms in both $\mathcal{LOCAL}$ and $\mathcal{CONGEST}$ ran on the vertices (rather than the edges, which represents communication lines) the color of each vertex should be known only to the vertex itself. On the other hand, in edge coloring, both edge endpoints are responsible for the coloring of the edge, which means that in terms of privacy preserving we may consider the edge coloring as private when at most the two endpoints of each edge know the color of the edge. However, when the graph is directed, we may demand that only the source endpoint of the edge will be aware of edge's color. \par
Nevertheless, there is a strong connection between graph vertex coloring to edge coloring. The similarity between the two problems is obvious, but more interestingly, there is a straight reduction between vertex coloring and edge coloring algorithms for general graphs. In the following section we will use this reduction in order to perform privacy preserving generic edge coloring of graphs. This reduction can be used to apply the defective graph coloring we presented in section \ref{section:poly-defective-coloring} in order to compute a defective edge coloring. There is however, a faster way to get a defective edge coloring. This technique, which is due to \cite{kuhn2009weak} will be presented in the later section. \par

\subsubsection{Edge Coloring Using Line Graphs}
Given a graph $G=(V,E)$, a \textit{line graph} $L(G)$ of graph $G$ is a graph which is constructed as follows. Each edge $e \in E$ becomes a vertex of the line graph $L(G)=(E,\mathcal{E})$. Each two distinct vertices of the line graph $e_1, e_2 \in E$ are connected ($(e_1,e_2)\in \mathcal{E}$) if they are incident to a single vertex in the original graph, i.e. there exist three vertices $v,u,w\in V$ such that $e_1=(v,u)$ and $e_2=(v,w)$. Observe that the line graph has $m \leq n^2$ vertices and a maximum degree of $\Delta(L(G)) = (2\Delta(G)-1)$. Also observe that a legal vertex coloring in the line graph $L(G)$ is a legal edge coloring in the original graph. Hence, if any vertex of the original graph is responsible for the coloring of part of its incident edges, the vertices can produce a legal graph coloring for the line graph and translate it to a legal edge coloring of the original graph. The assignment of each vertex to any incident edge can be done by specifying that for any edge $(u,v) \in E$, the vertex with the greater ID is responsible for the coloring of the edge in the line graph. \par
As a result, the algorithms provided in section \ref{section:generic_graph_coloring} can be applied to the line graph in order to produce a generic edge coloring. Given a graph $G=(V,E)$ with maximum degree $\Delta$, the algorithm for $2\Delta c\cdot \log{n}$-Coloring, applied on the line graph $L(G)$, produces an $2\cdot (2\Delta -1)\cdot c \cdot \log{(n^2)}=8\Delta \cdot c \cdot \log{n}$-edge-coloring, with solution domain of size $c \cdot \log{(n^2)}/2=c \cdot \log{n}$, which yields a contingency factor of $O(\Delta)$. The algorithm for $O(\Delta^2)$-Coloring, applied on the line graph $L(G)$, produces an $O((2\Delta -1)^2)=O(\Delta^2)$ edge coloring of the original graph, with a solution domain of size $(2\Delta-1)$, which keeps a contingency factor of $O(\Delta)$.

\subsubsection{Generic Defective Edge Coloring}
The line graph, which was presented in the previous section, can be used in order to perform a $p$-defective $O\left(\left(\frac{2\Delta-1}{p}\right)^2\right)$-generic edge coloring of the original graph by applying the algorithm from Theorem \ref{theorem:poly_defective} on the line graph. Such an algorithm will achieve a solution domain of size $O(\Delta/p)$ and a contingency factor of $O(\Delta/p)$ as well, both in $O(\log^*{n})$ rounds. There is however a faster privacy preserving algorithm for $p$-defective $O\left(\left(\frac{2\Delta-1}{p}\right)^2\right)$-defective coloring based on the algorithm of Kuhn (Algorithm 3 in \cite{kuhn2009weak}). \par
Kuhn's algorithm goes as follows. Suppose we have an undirected graph $G=(V,E)$ with maximum degree $\Delta$, and a constant $i\geq 1$. Each vertex numbers its adjacent edges with numbers between $\{1,...,\lceil\Delta/i\rceil\}$ such that each number will be assigned to at most $i$ of the vertex's adjacent edges. Then each vertex sends the number of each of its adjacent edges to the vertex on the other endpoint of the edge. Suppose that for a graph $G=(V,E)$, and an edge $(u,v) \in E$, $e_u$ and $e_v$ are the colors that was assigned to edge $e$ by vertex $u$ and $v$ (respectively). The set $\{e_u,e_v\}$ is assigned to be the color of edge $e$. Kuhn showed that this simple $O(1)$ rounds algorithm achieves a $4i-2$-defective ${\lceil\Delta/i\rceil+1} \choose {2}$-Edge Coloring. Setting $p=4i-2$ we get an $p$-defective $O\left(\left(\frac{\Delta}{p}\right)^2\right)$-edge coloring. \par
Kuhn's algorithm performs communication only between the two endpoints of each edge and only once. Hence the knowledge about the color of each edge is held only by its two endpoints. While in Kuhn's algorithm the nodes do not assign the labels randomly, in our algorithm a random assignment is required. Therefore, our algorithm preserves privacy. Since every vertex assigns the numbers independently, each of the available colors may be assigned (in certain scenario) to each edge. As a result, we achieve the following theorem. \par
\begin{theorem}
There is a privacy preserving algorithm for $p$-defective $O\left(\left(\frac{\Delta}{p}\right)^2\right)$-edge coloring with contingency factor of 1.
\end{theorem}

\subsection{Generic Algorithm for Edge Dominating Set Colored with $O(\sqrt{\Delta})$ Colors}
Given a graph $G=(V,E)$, an \textit{Edge Dominating Set} is a set of edges $D\subseteq E$ such that for every edge $e \in E \setminus D$ there is an edge $e'\in D$ which is adjacent to $e$. \par
In this section we will use coloring and maximal matching algorithms in order to produce an edge dominating set which is properly colored with $t\cdot \sqrt{\Delta}$ colors (for any constant integer $t>1$). While both the maximal matching and the minimal edge dominating set problems do not suit our privacy notion (due to their bipartitional nature), this partial-coloring problem variant can be calculated efficiently while still preserve privacy using special attributes of maximal matchings and edge dominating sets. While the identity of each member of the dominating set does not preserve privacy, its color is still private. As a result, the algorithm may be useful for a variety of network problems which wish to achieve privately a small amount of labels in the expense of achieving a partial coloring. Nevertheless, each uncolored edge is adjacent to at least one colored edge, from a reduced palette, which is of interest. \par
The outline of our algorithm is described in algorithm \ref{algo:dominating_edge_coloring}.\par
\vspace{5pt}
\begin{algorithm}[H]\label{algo:dominating_edge_coloring}
\caption{GENERIC\_EDGE\_DOMINATING\_COLORING}
\SetAlgoLined
\KwResult{A dominating set $D\subseteq E$ and a set of $t\cdot \sqrt{\Delta}$ colors for each edge $e \in D$ }
All the vertices run in parallel an $(c\cdot \Delta)$-Edge Coloring algorithm (for some constant $c>1$) which results in the coloring function $\varphi :E\rightarrow [c\cdot\Delta]$.\\
 Every edge is assigned to a distinct set $E_1,E_2,...,E_{\sqrt{\Delta}}$ according to its color such that for every $1\leq i \leq \sqrt{\Delta}$ and any $e \in E$, $e\in E_i$ iff $(i-1)\cdot \frac{c\cdot\Delta}{\sqrt{\Delta}} \leq \varphi(e) < i\cdot \frac{c\cdot\Delta}{\sqrt{\Delta}}$.\\
 Every set of edges $E_i$ computes independently a maximal matching $M_i$ on the graph formed by its edges. \\
 For every $1\leq i\leq \sqrt{\Delta}$, the set of colors $\{(i-1)\cdot t,(i-1)\cdot t +1,...,i\cdot t -1\}$ (for some constant $t>1$) is assigned as possible valid colors to each edge of set $M_i$ \\
\end{algorithm}
\vspace{5pt}
Given a graph $G=(V,E)$ let $\Gamma(e)$ ($e=(u,v)\in E$) be the set of the edges incident to $e$ in graph $G$, i.e. $\Gamma(e)=\{e'\in E \mid \exists{w\in V} \mathrm{\ s.t.\ } (v,w)\in E \mathrm{\ or\ } (u,w)\in E \}$. First, we analyze the correctness of algorithm \ref{algo:dominating_edge_coloring}:
\begin{lemma}
The set of edges $\bigcup\limits_{i=1}^{\sqrt{\Delta}}{M_i}$ is an edge dominating set of graph $G$.
\end{lemma}
\begin{proof} Assume, for contradiction that there exists an edge $e\in E,\ e\notin\bigcup\limits_{i=1}^{\sqrt{\Delta}}{M_i}$ such that for every edge $e'\in \Gamma(e)$, $e'\notin\bigcup\limits_{i=1}^{\sqrt{\Delta}}{M_i}$. Let $j$ be an integer such that $e\in E_j$. Since none of the edges in $\Gamma(e)$ are part of any maximal matching, it means that $e$ can add itself to the maximal matching $M_j$ without affecting the validity of the matching. Hence, $M_j$ is not maximal, which is a contradiction. 
\end{proof}
The coloring is done by each vertex picking privately in random a color out of its set of valid colors.
\begin{lemma}
Any coloring selection is a valid $(t\cdot\sqrt{\Delta})$-edge coloring of the dominating set.
\end{lemma}
\begin{proof} 
Since each maximal matching gets its own distinct range of valid colors, a conflict (where two adjacent edges chose the same color) may happen only within the edges of the same set of maximal matching $M_i$ (for any $1\leq i\leq \sqrt{\Delta}$). But since $M_i$ is a valid matching, it does not contain any adjacent edges. Hence, there is no pair of adjacent edges which were colored by the same color.
\end{proof}
Now, we will analyze the round complexity of the algorithm. Computing a valid $(c\cdot \Delta)$-edge coloring of graph $G$ can be done in $O(\log^3{\log{n}}\cdot \mathrm{polylog}({n}))$ rounds by \cite{harris2018distributed}. A maximal matching can be calculated in $O(\log^3{\log{n}}+\log{\Delta})$ by a combination of the randomized algorithm of \cite{barenboim2016locality} with the deterministic algorithm of \cite{fischer2017improved}. Hence, the overall rounds complexity is $O(\log^3{\log{n}}\cdot\mathrm{polyloglog}(n)+\log^3{\log{n}}+\log{\Delta})=O(\log{\Delta}+\log^3{\log{n}}\cdot\mathrm{polyloglog}(n))$. \par
The size of the problem domain is $t\cdot \sqrt{\Delta}$ and the size of the solution domain is $t$. Hence, the contingency factor is $\sqrt{\Delta}$. Consequently we achieve the following Theorem:
\begin{theorem}
Given a graph $G=(V,E)$ one can compute a valid edge dominating set with a generic $(t\cdot \sqrt{\Delta})$-edge coloring in $\tilde{O}(\log{\Delta}+\log^3{\log{n}})$ rounds with contingency factor of $\sqrt{\Delta}$.
\end{theorem}

\section{Conclusion}
\label{chap:conclusion}

The computer-science research fields of secure multi-party computation and distributed algorithms were both highly investigated during the last decades. While both of the fields prosper and yield many theoretical and practical results, the connection between these fields was made only seldom. Nevertheless, as implementation of distributed algorithms becomes common in sensor networks and IoT (Internet of Things) architectures, efficient privacy preserving techniques are essential. \par
In this work we present a novel approach, which rather than turning existing algorithms into secure ones, identifies and develops those algorithms that are inherently secure. Naturally, our work focuses on labeling problems. The inherently secure algorithms analyzed in this work are listed in 

We believe that these results establishes a broad basis for further research of both inherently secure algorithm and efficient techniques to translate distributed algorithms into secure algorithms. Such algorithms will open new possibilities for secure interconnection between machines, eliminating the need to mediate through a central secure server. As a consequence, distributed communication would possibly open a free and secure way to transmit data and solve problems.

\subsubsection*{Acknowledgements}
We are grateful to Merav Parter and  Eylon Yogev for helpful remarks.

\bibliographystyle{plainurl}% the mandatory bibstyle
\bibliography{acmart}

\newpage
\appendix
\section{Detailed Definitions of Graph Theory Problems}
\label{app:problem-def}
\subsection{Labeling Problems}
\textbf{Graph Coloring: }
The \textit{graph coloring} (or \textit{vertex coloring}) problem has its origins in the 19th century. Throughout the last few decades solutions for the graph coloring problem became the basis of many resource allocation and scheduling algorithms. \par
Formally, a function $\varphi :V \rightarrow [\alpha]$ is a legal \textit{$\alpha$-Coloring} of graph $G=(V,E)$ if and only if, for each $\{v,u\}\in E \rightarrow \varphi(v) \neq \varphi(u)$. The chromatic number of graph $G$, $\chi(G)$, is the minimum $\alpha$ such that a legal $\alpha$-coloring $\varphi:V\rightarrow [\alpha]$ exists. \par
While finding and validating the chromatic number requires knowledge about the topology of the entire graph, validation of $(\Delta+1)$-coloring requires knowledge only about the local neighborhood of each node. Furthermore, a construction of a valid graph coloring with $\alpha \geq \Delta + 1$ colors, can be achieved with less than $O(D(G))$ rounds. Nevertheless, finding the optimal algorithm for $(\Delta+1)$-coloring is still a challenge while deterministic construction of such a coloring requires at least $O(\log^*{n})$ rounds \cite{linial1987distributive}. A randomized $(\Delta+1)$-coloring can be done in $\mathrm{poly}(\log\log{n})$ rounds \cite{rozhovn2020polylogarithmic}. \par

\textbf{Edge Coloring: }
\textit{Edge coloring} is somewhat a complimentary problem to the graph coloring. While in graph coloring the coloring function assigns a color to each vertex, in the edge coloring we assign a color to each edge. Formally, a valid $\alpha$-edge coloring of graph $G=(V,E)$ is defined as a function $\varphi: E \rightarrow [\alpha]$ such that for any vertex $v \in V$ there are no two distinct vertices $u,w \in \Gamma(v)$ such that $\varphi((v,u)) = \varphi((v,w))$. For a graph with maximum degree $\Delta$, each of the endpoints of any edge $e=(u,v)$, is incident to up to $\Delta -1$ edges other than $e$. Hence, even if $e$'s $2\cdot (\Delta -1)$ incident edges have already selected distinct colors, a ground set of $(2\Delta-1)$ colors will be enough to achieve a valid edge coloring depending only on knowledge about the local neighborhood. Therefore, an $\alpha$-degree coloring with $\alpha \geq (2\Delta -1)$ is of special interest in the field of distributed algorithms. \par

\textbf{Forest Decomposition: }
A \textit{forest} is such a graph that all its connected components are trees. In other words, a forest is a graph which contains no cycles. A \textit{forest decomposition} of graph $G=(V,E)$ is an edge-disjoint partition of $G$, to $\alpha$ sub-graphs $\mathcal{F}_1,\mathcal{F}_2,...,\mathcal{F}_\alpha$ such that $\mathcal{F}_i$ is a forest for every $1 \leq i \leq \alpha$. Any graph $G$ with maximum degree $\Delta$ can be decomposed into $\Delta$ forests. One way to define the \textit{arboricity} of a graph $G$ is as the minimal number of forests which are enough to fully cover $G$. The arboricity is always at most the maximum degree ($\Delta$), and in many graph families is much smaller than the maximum degree. \par
Forest decomposition was proved to be a good building block for efficient distributed algorithms which solve other problems like coloring and MIS (such as in \cite{barenboim2010sublogarithmic}) especially for graphs with bounded arboricity. \par

\textbf{Network Decomposition: }
A graph with a $\mathrm{c}$-labeling of the vertices is \textit{strong} (respectively, \textit{weak}) \textit{$(\mathrm{d},\mathrm{c})$-network-decomposition}, if each connected component of vertices with the same label has \textit{strong} (respectively, \textit{weak}) diameter at most $\mathrm{d}$. Formally, given a graph $G=(V,E)$ and a vertex-disjoint partition of graph $G=(V,E)$ to $\alpha$ clusters $\mathcal{C}_1,\mathcal{C}_2,...,\mathcal{C}_\alpha$, we define an auxiliary graph $\mathcal{G}=(\mathcal{V},\mathcal{E})$ such that $\mathcal{V} = \{\mathcal{C}_1,\mathcal{C}_2,...,\mathcal{C}_\alpha\}$ and $(\mathcal{C}_u,\mathcal{C}_v)\in \mathcal{E}$ ($\mathcal{C}_u,\mathcal{C}_v\in \mathcal{V}$) iff $\exists (u,v) \in E$ such that $u\in\mathcal{C}_u$ and $v\in\mathcal{C}_v$.
The partition $\mathcal{C}_1,\mathcal{C}_2,...,\mathcal{C}_\alpha$ is a valid \textit{$(\mathrm{d},\mathrm{c})$-network decomposition} if (1) the chromatic number of $\mathcal{G}$ is at most $\mathrm{c}$ and (2) the distance between each pair of vertices contained in the same cluster $v,u\in \mathcal{C}_i$ is at most $\mathrm{d}$. In strong network decomposition, the distance is measured with respect to the cluster $\mathcal{C}_i$ (in other words, $dist_{\mathcal{C}_i}(v,u) \leq \mathrm{d}$). In weak network decomposition, the distance is measured with respect to the original graph $G$ (in other words, $dist_G(v,u) \leq \mathrm{d}$).\par 
Different algorithms yield different kind of network decompositions which satisfy different values of $\mathrm{d}$ and $\mathrm{c}$. One of the most valuable decompositions, which presents good trade-off between the radius of each cluster and the chromatic number of the auxiliary graph is an $(O(\log{n}),O(\log{n})$-network decomposition. \par
Network Decomposition is a valuable building block for algorithms to solve other graph theory problems. For instance, given a $(\mathrm{d},\mathrm{c})$-network decomposition one can solve MIS and $(\Delta+1)$-coloring in $O(\mathrm{d}\cdot \mathrm{c})$ rounds \cite{barenboim2011deterministic}. The basic idea is that each cluster learns its own topology in $O(\mathrm{d})$ rounds, and then solves any problem on the local cluster. Later, all the cluster-local solutions are collected to form an overall solution for the entire graph. \par

\subsection{Bipartition Problems} 
\textbf{Maximal Independent Set (MIS): }
A set $I \in V$ of vertices is called an \textit{Independent Set} (\textit{IS}) if and only if for each pair of vertices $v,u \in I$ there is no edge $(v,u) \in E$. A \textit{Maximal Independent Set} (\textit{MIS}) is an independent set $I$ that is maximal with respect to addition of vertices, i.e. there is no vertex $v \in V \setminus I$ such that $I \cup \{v\}$ is a valid independent set. \par

\textbf{Maximal Matching (MM): }
A set of edges $M \in E$ is called a \textit{Matching} if and only if every vertex $v \in V$ is connected to at most one edge in $M$, i.e. there is no pair of vertices $u_1,u_2 \in V$ ($u_1 \neq u_2$) such that $\{(u_1,v),(u_2,v)\} \subseteq M$. A \textit{Maximal Matching} (\textit{MM}) is a matching that is maximal with respect to addition of edges, i.e. there is no edge $e \in E \setminus M$ such that $M \cup \{e\}$ is a valid matching. \par

\end{document}